\documentclass[twocolumn,showpacs,preprintnumbers,amsmath,amssymb]{revtex4}

\usepackage{tabularx}
\usepackage{array}
\usepackage{mathrsfs}
\usepackage{amssymb}
\usepackage{amsmath}
\usepackage{graphicx}% Include figure files
\usepackage{bm}% bold math
\usepackage[dvipdfm,pdfstartview=FitH,colorlinks,linkcolor=blue,anchorcolor=blue,citecolor=blue]{hyperref}
\usepackage{amsthm}  %It is important for the proof environment
\newtheorem{theorem}{Theorem}
\newtheorem{lemma}{Lemma}

\begin{document}

\title{Robust quantum random number generator based on avalanche photodiodes}

\author{Fang-Xiang Wang,$^{1}$ Chao Wang,$^{1}$ Wei Chen, $^{1}$\footnote{weich@ustc.edu.cn} Shuang Wang, $^{1}$\footnote{wshuang@ustc.edu.cn} Fu-Sheng Lv,$^{2}$ De-Yong He,$^{1}$ Zhen-Qiang Yin,$^{1}$ Hong-Wei Li,$^{1}$ Guang-Can Guo,$^{1}$ and Zheng-Fu Han,$^{1}$}
\address{ $^1$ Key Laboratory of Quantum Information, University of Science and Technology of China, Hefei 230026, China\\
and Synergetic Innovation Center of Quantum Information $\&$ Quantum Physics, University of Science and Technology of China,\\
Hefei, Anhui 230026, China\\
$^2$ Department of Mathematics and LPMC, Nankai University, Tianjin 300071, China \\}

%\date{\today}% It is always \today, today,
             %  but any date may be explicitly specified

\begin{abstract}

We propose and demonstrate a scheme to realize a high-efficiency truly quantum random number generator (RNG) at room temperature (RT). Using an effective extractor with simple time bin encoding method, the avalanche pulses of avalanche photodiode (APD) are converted into high-quality random numbers (RNs) that are robust to slow varying noise such as fluctuations of pulse intensity and temperature. A light source is compatible but not necessary in this scheme. Therefor the robustness of the system is effective enhanced. The random bits generation rate of this proof-of-principle system is 0.69 Mbps with double APDs and 0.34 Mbps with single APD. The results indicate that a high-speed RNG chip based on the scheme is potentially available with an integrable APD array.

\end{abstract}

%\pacs{42.50.Lc, 03.67.Ac}

%\pacs{03.67.Dd, 05.40.-a, 42.50.Lc}
%\keywords{random number genetator, avalanche photodiode, robust, unbiased, efficiency, integratable}

\maketitle

%\section{Introduction}
\section{Introduction}

{R}{andom} numbers are important in many fields of scientific research and real-life applications, such as fundamental physical research, computer science and the lottery industry. Although pseudo RNs can be generated by computer software and hardware with very high speed, high quality truly random number generators (TRNGs) must be adopted in some important applications. For example, TRNGs play important roles in information security, in which quantum cryptography is an emerging technology with potential applications to the next generation information security infrastructure.

The unpredictability of a physical procedure is the resource for generating truly RNs, and two steps are typically necessary to generate RNs with these procedures. First, signals related to a random physical procedure must be effectively generated and gathered. There have been many TRNGs based on physics, for example, circuit noise \cite{Holman,Petrie,Jun,Bucci} and radioactive decay \cite{Isida}. In all available elements, quantum mechanics is good for generating a nondeterministic signal. Some RNGs are designed with a quantum procedure in nature, such as wave function collapse of  single photon due to measurement \cite{Stefanov2000,Jennewein,Ma2005,Furst,Ren}, entangled state measurement \cite{Pironio}, effects of vacuum fluctuation \cite{Gabriel,Jofre} and quantum phase fluctuation \cite{Xu}. In many of these schemes, an almost single photon light source is necessary for generating quantum signals \cite{Jennewein,Furst,Ren,Dynes}. A quantum random number generator (QRNG), in which the light source is not necessary but compatible may have advantages in integration and usage.

The second step of a physical RNG is to implement an effective encoding method to transform these physical signals into RNs. The efficiency of encoding methods is a key limitation for the generation rate of RNGs. As devices and environments may vary in real time, postprocessing will be necessary to generate high-quality RNs. Even commercial products, like IdQuantique Quantis random number generator, for which Photons - light particles - are sent one by one onto a semi-transparent mirror and detected, cannot avoid postprocessing. Algorithms applied to raw RNs may reduce the efficiency of the final RNs and increase the complexity and cost of the system. Thus, the kernel of a TRNG is an effective, encoding method that is immune to slowly varying noise interference and is no need for complex postprocessing to remove bias. What should be noted here is the boundary between encoding methods and postprocessing algorithms. Although the boundary is not clearly defined, we adopt the principle that postprocessing algorithms take a large amount of resources \cite{Mario}. Many QRNGs can generate high-quality RNs by utilizing simple encoding methods, but efficiency is a dominant limitation for most ones \cite{Ren,Dynes,Wei}, for example, the efficiency in reference \cite{Ren} is 40$\%$. Some other one achieved very high rates by encoding the amplitude of the probe current of the detector into multi-bit RNs \cite{Liu}. However, the amplitude of the probe current is sensitive to devices and environments, so stable devices with high resolution are required when implementing these RNGs, which indicates higher cost and greater complexity.

Bias-free physical processes are perfect for RNG so that the encoding method will be mostly simple and will use minimal resources. However, processes used in RNG are always biased, so that the encoding method plays a key role in the RNG to obtain %unbiased 
high-quality RNs. John von Neumann first proposed an unbiased encoding method for biased Bernoulli trials \cite{Neumann}. It has been used in QRNG \cite{Wei}. However, the efficiency limitation of the von Neumann method is 0.25. The method was subsequently developed in order to obtain a higher efficiency \cite{Elias,Hoeffding,Samuelson}, among which Elias produced a very high efficiency for infinite situations \cite{Elias}. Ren et al. proposed another encoding method based on the precise discrimination of photon numbers of two consecutive pulses \cite{Ren}.This scheme needs high precision devices to discriminate photon numbers, and the efficiency limitation of the method is 0.5.

In this study, we propose a TRNG scheme based on the discrimination avalanche pulses of APD. These pulses can be generated by the dark current of APD or incident photons, so that a light source is compatible but not necessary in the scheme. A robust encoding method for biased Bernoulli trials with higher efficiency than previous works is proposed. Furthermore, we test the system with multi-APDs, the experiment results indicate the feasibility of implementation of high-quality, robust QRNG chips using an integrated APD array.

\section{RNG scheme}
\label{RNG scheme}

According to the quantum theory of lasers, the photon statistics of a laser pulse operating above threshold follows the Poisson distribution \cite{Scully}, which can be preserved after drastic attenuation. The Poisson distribution is
\begin{equation}
P_{\lambda}(\textit{n})
=
\frac{\lambda^n}{n!}e^{-\lambda},
\end{equation}
where $\lambda$ is the mean photon number of a laser pulse and $P_{\lambda}(\textit{n})$
is the probability that the pulse contains $n$ photons. Thus, the coherent state produces an unpredictable photon number for every detection, and this quantum property can be used to implement QRNG.

If the detection efficiency of the APD is not taken into account, the probability of an avalanche pulse caused by a laser pulse produces a photon number $n > 0$ in the pulse, which can be described as
\begin{equation}
\sum_{n>0} P_{\lambda}(\textit{n})=1-P_{\lambda}(n=0)=1-e^{-\lambda}.
\end{equation}

The avalanche pulses of an APD can be generated by dark currents. Because of thermal fluctuation, electrons of the APD may transit from the top of the valence band to the conduction band. Electrons in the conduction band are sped up by the high reverse-bias electric field and lead to avalanche pulses. Because the thermal fluctuation at RT is much smaller than the energy gap between the valence band and the conduction band, the transiting probability is very small. We use the tight-binding approximation here. Thus, the transiting events of different atoms are independent identically distributed (IID), and the statistics of total events follow the Bernoulli distribution
\begin{equation}
P(n_1=k)={N_1\choose k}\ p_1^k\ (1-p_1)^{N_1-k}
\end{equation}
where $n_1$ is the total electron number transiting to the conduction band during a certain priod $\tau,N_1$ is the total electron number at the top of valence band, $p_1$ is the transiting probability of a single electron, and $P(n_1=k)$ is the transiting probability of k electrons. As $p_1$ is much smaller than 1 and $N_1$ is large in the material, the transiting process follows the Poisson limit theorem
\begin{equation}
\lim_{N_1 \to \infty , p_1 \to 0} {N_1\choose k} p_1^k{(1-p_1)}^{N_1-k}=
\frac{\lambda_1^k}{k!}e^{-\lambda_1}=P_{\lambda_1}(k),
\end{equation}
where $\lambda_1$ is the mean electron number transiting to the conduction band during $\tau$ . Thus, the electrons' transiting process follows the Poisson distribution, and so does the dark count. The dark count is then as usable as a laser pulse. Dark counts of APD were first used by Tawfeeq to propose a RNG scheme \cite{Tawfeeq}. The scheme was easier and provided a new idea regarding RNG based on APD. However, that scheme did not show an effective encoding method and could not generate RNs. And RNG with high randomness based on dark counts of APD has not been implemented before.

Because the sum of Poisson-distributed random variables follows the Poisson distribution, the sum probability of an avalanche during $\tau'$ follows
\begin{equation}
p=P_{\lambda^{'}} (n>0)=1-e^{-\lambda'},
\end{equation}
where $\tau^{'}$ is the detection window, $\lambda^{'}=\eta (\lambda+\lambda_1)$ , and $\eta$ is the detection efficiency of APD. Clearly, the probability of no avalanche pulse is $q=1-p=e^{-\lambda^{'}}.$ The detection process is then a Bernoulli trial. A simple and robust encoding method for biased Bernoulli trials is then proposed here. It is an extension of the von Neumann method but with much higher efficiency. The encoding method is constructed as follows.

We consider the physical system of an avalanche photodiode (APD) working on the Geiger mode. We treat a detection window of APD as a time bin and sequence these time bins with time. According to the discussion above, avalanches caused by laser and thermal fluctuation in different time bins are IID if experimental parameters are constant, namely, $p=p_0$. In fact, experimental parameters are hardly constant. We claim that our encoding method constructed here also applies to the condition that experimental parameters vary slowly so that the encoding method is effective and robust. We mark a "1" in a time bin if an avalanche happens in the corresponding detection window; otherwise we mark a "0" in it. Considering $N$ time bins happened successively as a time-bin block. There are totally ${N\choose k}$ possible combinations when $k$ "1" are marked in the block if we do not get additional information about the block, namely, the uncertainty of these $N$ time bins are $N\choose k$. These equiprobable $N\choose k$ possible combinations are then encoded into uniform RNs from 0 to ${N\choose k}-1$. The encoding processes are one-to-one mapping and the mapping function is
\begin{equation}
f(k_1,k_2,\cdots,k_k)=\sum_{j=1}^{k}{N-k_j\choose k-j+1},
\end{equation}
where, $k_j$ means that the $j$-th "1" happened in the $k_j$-th time bin.

Then, we go to the interpretation of the mapping function. As discussed, if we only know that there are $k$ "1" in the time-bin block, the uncertainty is $N\choose k$. If we get the temporal information in the time bin suquence of the first "1", namely, $k_1$ is known, the uncertainty reduces to ${N-k_1 \choose k-1}$, in other words, the information content we get is ${N\choose k}-{N-k_1 \choose k-1}$. As similar, when $k_2$ is also known, the information content we get increases by ${N-k_1 \choose k-1} - {N-k_2 \choose k-2}$. The uncertainty will reduce further if $k_3, k_4, \cdots$ are also known. In the extreme case, if $k_1, k_2, \cdots, k_k$ are all know, the uncertainty remaining becomes zero, and we get all information content about these $N\choose k$ combinations. We sum all information content got with $k_1$ to $k_k$ and the summation is the RN we want. The mapping function is $f(k_1,k_2,\cdots,k_k)=\sum_{j=1}^{k}{N-k_j+1\choose k-j+1}-{N-k_j\choose k-j}=\sum_{j=1}^{k}{N-k_j\choose k-j+1}$, where we have used the combination formula ${N+1\choose k+1}={N\choose k+1}+{N\choose k}$. It is evident that $N-k_j\choose k-j+1$ is monotonic with $k_j$. Thus, the mapping function is monotonic. The maximum possible number got from mapping function is $f(k_1=1,k_2=2,\cdots,k_k=k)=\sum_{j=1}^{k}{N-j\choose k-j+1}={N\choose k}-1$. And the minimum possible number is $f(k_1=N-k+1,k_2=N-k+2,\cdots,k_k=N)=\sum_{j=1}^{k}{N-(N-k+j)\choose k-j+1}=0$. Thus, the encoding process is one-to-one mapping and the RN is in $N\choose k$ representation.

Taking into account the wide applications of binary RNs, the $N\choose k$-ary encoding method can go further and be modified by the binary method proposed by Elias \cite{Elias}. The method expands $N\choose k$ into subblocks as follows:
\begin{equation}
{N\choose k}=\alpha_m2^m+\alpha_{m-1}2^{m-1}+\cdots+\alpha_02^0.
\end{equation}
So that $\alpha_m,\alpha_{m-1},\cdots,\alpha_0$ are binary expansion coefficients of integer $N\choose k$, where $\alpha_m=1,\alpha_i=0$ or 1 for $0\leq i<m$. The subblock related to the $\alpha_0$ term should be abandoned, as it contains either one or no member and could not be encoded into RNs. Suppose the non-zero binary expansion coefficients are ${\alpha_m,\alpha_{i_1},\alpha_{i_2},\cdots,\alpha_{i_l}}$. If $f(k_1,k_2,\cdots,k_k)< 2^m$, convert $f(k_1,k_2,\cdots,k_k)$ into a m-bit binary number directly. If $2^m+\sum_{s=1}^{r}2^{i_s}\leq f(k_1,k_2,\cdots,k_k)<2^m+\sum_{s=1}^{r+1}2^{i_s}$, then convert $f'(k_1,k_2,\cdots,k_k)=f(k_1,k_2,\cdots,k_k)-(2^m+\sum_{s=1}^{r}2^{i_s})$ into a $i_{r+1}$-bit number directly (the schematic graph of encoding process is shown in Figure \ref{fig:wangf1}). $f(k_1,k_2,\cdots,k_k)$ is abandoned if $i_{r+1}=0$.
\begin{figure}
\centering
	\resizebox{8.5cm}{5cm}{\includegraphics{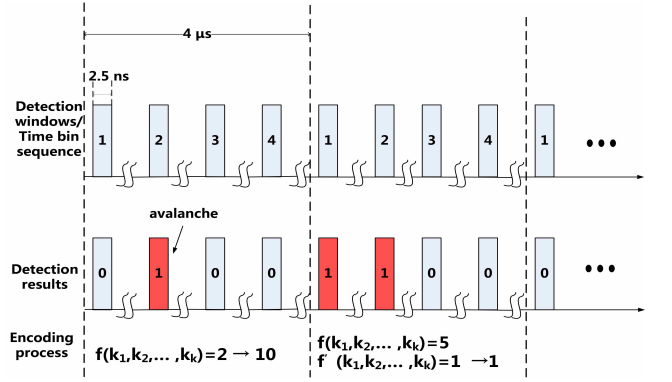}}
	\caption{(Color online) The schematic graph of encoding process in time sequence, where $N=4$. For the first $N$ Time bins (detection windows), $k=1, {N\choose k}=4=2^2$, $f(k_1,k_2,\cdots,k_k)<2^2$ always holds and $f(k_1,k_2,\cdots,k_k)=3$ converts into a 2-bit number "10" directly. For the second $N$ Time bins, $k=2, {N\choose k}=6=2^2+2^1$,$2^2<f(k_1,k_2,\cdots,k_k)=5<2^2+2^1, f'(k_1,k_2,\cdots,k_k)=5-2^2=1$, thus, $f'(k_1,k_2,\cdots,k_k)$ converts into a 1-bit number "1" directly.}
	\label{fig:wangf1}
\end{figure}

The encoding method constructed requires $p$ to be constant among these $N$ time bins (detection windows) of the same block. But the p values in different blocks are not necessarily identical, so that the method is robust to environment noise. As detection interval of APD can be as less as $\sim ns$, it is only $\sim \mu s$ when $N\sim 100$. It is reasonable to consider that parameters, depending on environments, which are slowly varying, are invariable in such a short time. These parameters can be laser intensity, temperature, etc. In additional, the raw RNs generated remain uniform even the slowly varying interferences are periodic (see Section \ref{Experimental setup} and Section \ref{Results and discussion}). Thus, the encoding method is effective, efficiency and robust in practice.

The encoding method is effective for any $k\neq 0,N$. Taking into account all possible value of $k$, the average encoding efficiency per time bin before binary expansion is
\begin{equation}
H(N,p)=-\frac{1}{N} \sum_{k=1}^{N-1}{N\choose k} p^k(1-p)^{N-k}(log_2\frac{1}{{N\choose k}}).
\end{equation}
A higher $H(N,p)$ indicates a higher efficiency. For any integer $N\geq 2$, the optimal $p$ for the average encoding efficiency $H(N,p)$ is $\frac{1}{2}$, and $H(N,p)\to S(p)$ as $N\to \infty$ (see Figure \ref{fig:wangf2}), where $S(p)$ is the Shannon entropy of a single Bernoulli trial. In addition, $H(N,\frac{1}{2})$ increases with $N$ and converges to 1 (the projection in circular blue curve of Figure \ref{fig:wangf2}). For instance, $H(5,\frac{1}{2})=0.5604$, while $H(10,\frac{1}{2})=0.7294$, the efficiency is much higher than previous ones based on single photon discrimination \cite{Ren,Dynes,Wei}. The theorems are proven in the Appendix. %\ref{Appendix}.

\begin{figure}
\centering
%\resizebox{16cm}{6cm}
\includegraphics[width=8cm,height=5cm]{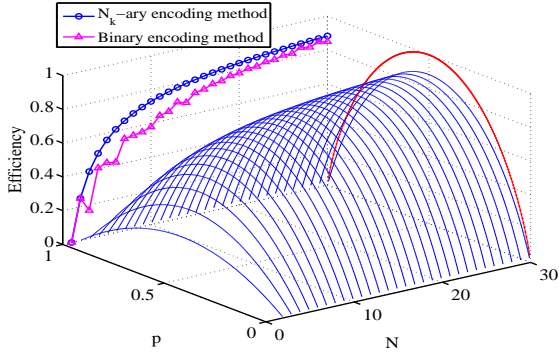}

\caption{(Color online) The average encoding efficiency per time bin increases with $N$. $H(N,p)$ (the 3-Dimensional blue curve) converges to $S(p)$ (the dotted red curve) with infinite $N$. The projection on the left side are the efficiencies of $N_k$-ary (the circular blue curve) and binary (the triangular pink curve) encoding methods for different $N$ when $p=\frac{1}{2}$. The projection shows that the encoding efficiency will converge to 1 with infinite $N$ when $p=\frac{1}{2}$ . The corresponding efficiency after the binary expansion is lower ($N>2$) but will converge to the $N_k$-ary one at large $N$.}
\label{fig:wangf2}
\end{figure}

The more subblocks $N\choose k$ divides into, the fewer possible combinations and thus the less uncertainty of the subblock there will be. In addition, the uncertainties among different subblocks (blocks) are not utilized in both encoding methods above. Thus, more subblocks indicate more uncertainty among subblocks and hence less extracted entropy and encoding efficiency, as the total entropy is conserved. The output sequences in binary representation are therefore obtained at the cost of entropy or efficiency, and the efficiency becomes
\begin{equation}
\begin{aligned}
& H_b(N,p)=\frac{1}{N}\sum_{k=1}^{N-1} p^k(1-p)^{N-k}(\alpha_{m_k}2^{m_k}m\\
&\qquad\qquad\qquad\qquad+\alpha_{{m-1}_k}2^{{m-1}_k}(m-1)\\
&\qquad\qquad\qquad\qquad+\cdots+\alpha_{0_k}2^{0_k}\cdot 0),
\end{aligned}
\end{equation}
 where the subscript $k$ means there are $k$ "1" in the block. The efficiency after expansion is shown in Figure \ref{fig:wangf2} (the projection in triangular pink curve). Moreover, more blocks mean more resources to be required.

Afterpulsing effect will lead to bias of IID events above. Its influence on QRNG will be discussed in Section \ref{Results and discussion}. It should be note that a similar spatial encoding method has been proposed for a different physical system by Marangon et.al.\cite{Marangon2014}.

\section{Experimental setup}
\label{Experimental setup}

\begin{figure}
\centering
\resizebox{8.5cm}{3cm}{\includegraphics{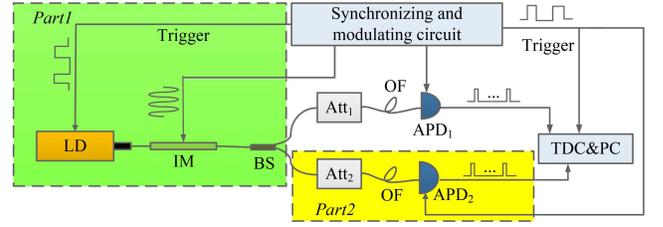}}
\caption{(Color online) Schematic setup of the experiment. LD: Laser diode; IM: optical intensity modulator; BS: beam splitter; Att: electronically variable optical attenuators (EVOA); OF: optical fiber; APD: avalanche photodiode; TDC: time-to-digital converter; PC: personal computer}
\label{fig:wangf3}
\end{figure}

Three scenarios were designed in order to evaluate the feasibility and the robustness of this scheme. \textbf{(a)} Avalanche singles from a single APD was acquired and data were encoded according to the method of Section \ref{RNG scheme} in order to verify if the method can generate high-quality raw RNs. A laser diode (LD) was added in this setup as an optional light source to increase the RN generation rate. By modulating the power of LD, we simulated an additional noise and the variation of the avalanche efficiency, so that the robustness of scheme was evaluated. \textbf{(b)} We added an additional APD to the system of setup (a). Two APDs were grouped, and the avalanche pulses were gathered and processed parallelly to generate RNs. This setup was to evaluate the possibility to increase the RN generation rate with APD arrays while keeping the high-quality feature of RNs. \textbf{(c)} In order to demonstrate the scheme can work properly without light source, the LD of setup (b) was removed and the avalanche pulsed were generated only by dark counts of APDs.

The system diagram is shown in Figure \ref{fig:wangf3}. A pulse LD with the wavelength of 1550 $nm$ was used as an optional light source and was trigged by 1 MHz electronic pulses. An intensity modulator (IM) following the LD was used to modulate the power of light pulses from LD. The output light pulses from IM were divided into two parts by a beam splitter (BS) and were attenuated to the single photon level by two electronically variable optical attenuators (EVOA). Then light pulses of different paths were coupled to two APDs (PGA-300, Princeton Lightwave), individually. The APDs worked in Geiger mode. The trigger frequency was 1 MHz and the gate width was 2.5 $ns$.

APDs used to detect single photons are commonly cooled from -30 ${}^{\circ}$C to -50 ${}^{\circ}$C in order to reduce the dark count rate, such as in quantum key distribution applications. In our experiments, dark count of APDs can be used as a resource to generate RNs as well as external photons. Thus, the cooling processes for APDs are not necessary, which makes the system more practical and less expensive. The APDs in our experiments were worked at RT (approximately 23 ${}^{\circ}$C).

The avalanche pulses of the two APDs were discriminated and amplified, then were sent into a time-to-digital converter (Agilent U1051A Acqiris TC890) to be processed. The TDC has one input channel of start signal and 6 input channels of stop signals, and can convert the time intervals between the stop and start signals into 32-bit numbers. The discrimination results from TDC were sequentially numbered with the time bin of 1 $\mu s$, and consequently encoded into RNs according to the method of Section \ref{RNG scheme} in real time. As a proof-of-principle experiment, the encoding process was executed every four successive detection windows. All trigger signals in the system were synchronized by a home-made circuit and the delay among them could be adjusted in the step of 10 $ps$, respectively.

The system could be divided into three parts, as shown in Figure \ref{fig:wangf3}. In Scenario (a), Part 2 was removed, thus single APD was used to generate RNs consequently. In Scenario (c), Part 1 was removed, and the RNs were exclusively generated by dark counts of APDs. We added a sinusoidal driving signal to the IM in Scenario (a) and Scenario (b) with frequency of 0.05 Hz and amplitude of 3 V in order to simulate an external noise. The average counting probabilities of the two APDs were initialized to 0.5 by adjusting the EVOAs ahead of them individually. According to the sinusoidal modulation of IM, the probabilities of the APDs varied from 0.3 to 0.7, which could be regarded as an external noise to APDs.

\section{Results and discussion}
\label{Results and discussion}

In our experiments, we set $N$ as 4 as discussed. The total 16 types of detection results were classified as 5 subsets according to the $k$ value. Two of these sixteen types of detection results, subsets with $k=0$ and $k=4$, were abandoned, while the other fourteen were set to generate RNs. \textbf{The generation rates} of Scenario (a) and Scenario (b) are functions of time $t$. The encoding efficiency after binary expansion of Scenario (a) is $H_{b}(N,p(t))=\frac{1}{20}\int_{0}^{20}H_{b}(N,p(t))\ dt$, where $p(t)=0.5+0.2sin(0.1\pi t)$. Substituting $N$ with $4$, we obtain $H_{b}(4,p(t))=0.3454$, and the corresponding generation rates is about $0.34$ Mbps. Scenario (b) has a double generation rate of about $0.69$ Mbps. As dark count of APDs used in Scenario (c) is relatively low, the encoding efficiency per APD is $H_{b}(4,0.01)=0.0197$ and the practical generation rate is about $0.04$ Mbps.

With current technologies, the generation rate of QRNG can be 100 Mbps to Gbps with the cost of utilizing stable and high resolution equipments \cite{Patel2012,Scarcella2015,Jofre,Xu,Yuan2014}. Although the generation rate of our proof-in-principle experiment is relatively much lower comparing with existing results, it can be remarkably increased with some measures. Firstly, using the APD array with high integration density can generate random bits concurrently with an acceptable cost growth, as has been demonstrated in our experiment. Secondly, the higher generation rate can be acquired with higher gating frequency of APD which can work exceed 2GHz \cite{Patel2012,Scarcella2015}. Thirdly, benefiting from the simple encoding method of the system, larger $N$ can be employed to improve the encoding efficiency while the random generation rate will not be restricted by the processing procedure of the raw key bits. 
\begin{figure}
\centering
\resizebox{8cm}{5cm}{\includegraphics{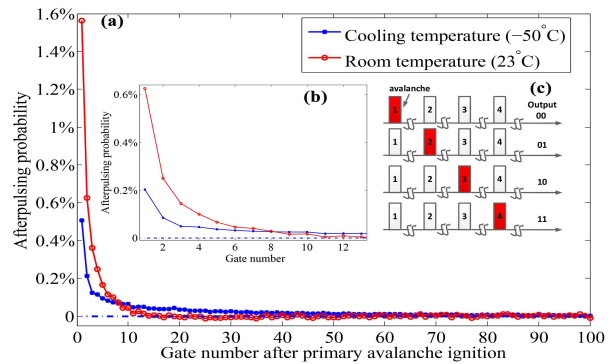}}
\caption{(Color online) (a) (b) Afterpulsing probability per gate (the gate frequency is 10 MHz) at room temperature (in red circle) and cooling temperature (in blue dot). (c) The schematic time sequence diagram of encoding process, where $N=4$ and $k=1$.}
\label{fig:wangf4}
\end{figure}

\textbf{Afterpulsing} is correlated to the primary avalanche \cite{Itzler2011}. Afterpulsing effects lead to bias of RNs generated. We compared afterpulsing probabilities of APD at cooling temperature (CT, $-50{}^{\circ}$) and RT, as shown in Figure \ref{fig:wangf4}. The data were measured under Geiger mode with 10 MHz gating frequency. The total afterpulsing probability at RT is 3.3\%, while it is 1.8\% at CT. According to Figure \ref{fig:wangf4}(a) (b), afterpulsing probability at RT is much larger for the top gates, but it decreases rapidly to zero and becomes smaller than that at RT after the eighth gate. 

Let $p_{a}(i)$ be the afterpulsing probability of the $i$-th gate after primary avalanche ignition. The probabilities of original IID Bernoulli trials are not equal any more and extracted entropy becomes less. For the case of $k=1$, showing in Figure \ref{fig:wangf4} (c), the probabilities of different events become
\begin{equation}
\begin{aligned}
& P_{k=1}(1)=p(1-p-p_a(1))(1-p-p_a(2))(1-p-p_a(3));\\
& P_{k=1}(2)=(1-p)p(1-p-p_a(1))(1-p-p_a(2));\\
& P_{k=1}(3)=(1-p)(1-p)p(1-p-p_a(1));\\
& P_{k=1}(4)=(1-p)(1-p)(1-p)p.
\end{aligned}
\end{equation}
where $P_{k=1}(i)$ represents the probability that the avalanche happens in the $i$-th time bin. For experiments here (the gating frequency is 1 MHz), $p_a(1)=4.3\times 10^{-4}, p_a(2)=p_a(3)=0,P_{k=1}(1)=0.062446,P_{k=1}(2)=0.062446,P_{k=1}(3)=0.062446,P_{k=1}(4)=0.0625$, where $p=0.5$ is adopted. The extracted entropy becomes $S=-\sum P_{k=1}(i)\log_{2}P_{k=1}(i)=2-10^{-7}$. The reduction of extracted entropy is negligible. When gating frequency is 10 MHz, the afterpulsing effect is more but not that significant. The corresponding extracted entropy reduces as less as $2.8\times 10^{-4}$ but the corresponding generation rate will be $6.9$ Mbps. For higher gating frequency and count rate, the afterpulsing probability is not a catastrophic problem. 2 GHz gating frequency and count rate as high as 650 Mcount/s InGaAs APDs have been realized, respectively \cite{Patel2012,Scarcella2015}. The afterpulsing probabilities, 4\% for reference \cite{Patel2012} and 1.5\% for reference \cite{Scarcella2015}, are still within the same level. For multi-avalanche case of high speed APD, the analysis of afterpulsing effect is much more complex and requires further study. 

We analyzed the \textbf{uniformity} of the RNs generated from these experiments to test the independence of different detection events. Only the uniformity of detection results of Scenario (a) are demonstrated, as shown in Figure\ \ref{fig:wangf5}, because the other two experiments contain very similar properties. Let $p(x,\ y)$ represent the population of elements of a $16\times 16$ matrix, which indicates that neighboring detection results $x$ and $y$ happen successively, and integers $x,\ y\in {0,\ 1,\ 2,\cdots,\ 15}$ represent the $16$ possible detection results for every four consecutive detections.

In Figure \ref{fig:wangf5}, different colors indicate different population $p(x,y)$. Clearly, the matrix is symmetric and is partioned into subblocks. The symmetric matrix shows that $p(x,y)=p(y,x)$ for any $x,y$. It means that there is no time correlation between successive $x$ and $y$. $p(x,y)$ of different elements in the same subblock are identical. It means the population is uniform in subblocks. Thus $p(x,y)=p(y,x)=p(x)p(y)$. This represents that each detection event is independent, and the imperfections of the beam splitter and detection efficiencies make no difference to the uniformity of the randomness extraction system. The experimental results are consistent with the theory we proposed in Section \ref{RNG scheme}.

\begin{figure}
\centering
\resizebox{8cm}{5cm}{\includegraphics{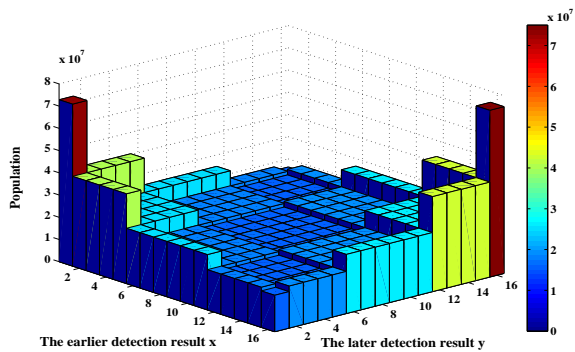}}
\caption{(Color online) The uniformities of RNs output from the RNG scheme. For every four detection windows, $x$ or $y=0$ indicates that no avalanche pulse is detected, $x$ or $y\in [1,4]$ indicates that one avalanche pulse is detected, $x$ or $y\in [5,10]$ indicates that two avalanche pulses are detected, $x$ or $y\in [11,14]$ indicates that three avalanche pulses are detected and $x$ or $y=15$ means four avalanche pulses are detected, where $x,\ y\in integer$. The altitude represents the population of detection results in which $x$ and $y$ happen successively.}
\label{fig:wangf5}
\end{figure}

\textbf{Min-entropy} evaluation was employed for all these experiments.Min-entropy, defined as
\begin{equation}
H_\infty=-{log}_2\{max\ p(x_i)\},
\end{equation}
is the evaluation of the worst situation, where $p(x_i)$ is the probability of possible output $x_i$, and $max\ p(x_i)$ is the maximal value of all $p(x_i)$. It is a strong way to measure the information content, while Shannon entropy is a weighted average evaluation. Both min-entropy and Shannon entropy are special cases of R\'enyi entropy \cite{Renyi,Bromiley}. Shannon entropy is the upper bound of min-entropy, and they coincide if and only if the distribution of the variable is uniform \cite{Cachin,Smith}. Min-entropy evaluation is therefore a good way to evaluate the quality of randomness of RNs.
\begin{figure}
\centering
\resizebox{8cm}{6cm}{\includegraphics{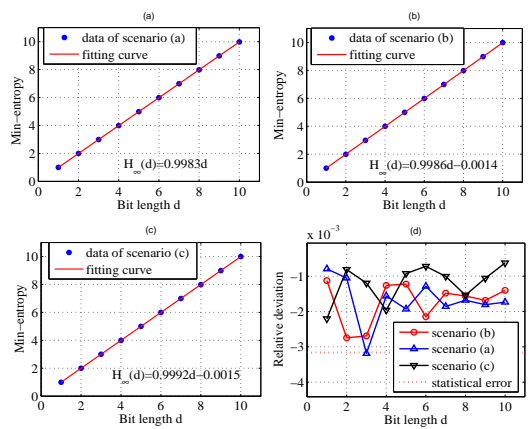}}
\caption{(Color online) (a) (b) (c) Min-entropy of samples (data point) output from scenarios (a), (b) and (c), respectively. The linear fitting function (fitting line) is shown. (d) The relative deviations between min-entropy and Shannon entropy of uniform distribution of all three scenarios.}
\label{fig:wangf6}
\end{figure}

The min-entropy of raw data from all scenarios were evaluated, where $i=0,1,\cdots,d-1$ and $p(x_0),p(x_1),\cdots,p(x_d-1)$ represent the probabilities of "0", "1",$\cdots$, "d-1" in binary representation respectively. As shown in Figure \ref{fig:wangf6}, the results show that deviations between min-entropy and Shannon entropy of uniform distribution are all of the order of 0.001. This is within the statistical error ($\sim 3.2\times 10^{-3}$, the red dot line in Figure \ref{fig:wangf6}(d)), as the statistical amount is $10^5\cdot 2^d$ bits, which indicates a good quality of randomness of the RNs.

The raw binary RN samples output by using the binary encoding method were tested using \textbf{NIST statistical test suit} \cite{http}. The standard statistical test suite, containing 15 subtests, calls for a 1-Gbit sample. The standard statistical test outputs two values, the p-value and the pass proportion, for each item. The sample passes the standard statistical test if and only if the p values and proportions of all items are larger than 0.0001 and 0.9805608, respectively. 20 samples for both Scenario (a) and Scenario (b) and 7 samples for Scenario (c) were tested by the standard statistical test. The testing results are shown in Table \ref{tab:1}. 2 samples of Scenario (a), 4 samples of Scenario (b) and 1 sample of Scenario (c) failed the statistical test. All failed samples only failed the NonOverlappingTemplate test. The focus of NonOverlappingTemplate test is the number of occurrences of pre-specified target strings. The purpose of this test is to detect generators that produce too many occurrences of a given non-periodic (aperiodic) pattern. The test outputs 148 group of p-values and pass proportions with different pre-specified 9-bit target strings. The sequence has irregular occurrences of the possible template patterns and fails the test if the p-value is smaller than the preset value. Only one of the 148 groups failed the preset value of pass proportion, 0.980, for every failed sample. The minimum value of pass proportions is 0.975, which is very close to 0.980. We cannot find out the reason of failures and suspect that afterpulsing effect is a potential candidate.

\begin{table*}
\centering
\caption{The standard statistical test results of NIST. Twenty samples of 1 Gbit were tested for %the first two experiments 
Scenario (a) and Scenario (b) and 7 samples for Scenario (c), as the dark count rate was lower. For the tests outputting multiple p values and proportions, the worst case was adopted.}
\begin{tabular}{ccccccc}
\hline \hline
 & \multicolumn{6}{c}{\bf Passed} \\
\cline{2-7}
\bf Testing item & \multicolumn{2}{c}{\bf Scenario (a)} & \multicolumn{2}{c}{\bf Scenario (b)} & \multicolumn{2}{c}{\bf Scenario (c)}\\
 & \bf \scriptsize Proportion & \bf \scriptsize p-value & \bf \scriptsize Proportion & \bf \scriptsize p-value & \bf\scriptsize Proportion & \bf \scriptsize p-value \\
\hline
\bf Frequency &	20/20 &	20/20 &	20/20 &	20/20 & 7/7 & 7/7 \\
\bf BlockFrequency & 20/20 & 20/20 & 20/20 & 20/20 & 7/7 & 7/7 \\
\bf CumulativeSums & 20/20 & 20/20 & 20/20 & 20/20 & 7/7 & 7/7 \\
\bf Run & 20/20 & 20/20 & 20/20 & 20/20 & 7/7	& 7/7 \\
\bf LongestRun & 20/20 & 20/20 & 20/20 & 20/20 & 7/7 & 7/7 \\
\bf Rank & 20/20 & 20/20 & 20/20 & 20/20 & 7/7	& 7/7 \\
\bf FFT & 20/20 & 20/20 & 20/20 & 20/20 & 7/7	& 7/7 \\
\bf NonOverlappingTemplate & 18/20 & 20/20 & 16/20 & 20/20 & 6/7 & 7/7 \\
\bf OverlappingTemplate & 20/20 & 20/20 & 20/20 & 20/20 & 7/7	& 7/7 \\
\bf Universal & 20/20 & 20/20 & 20/20 & 20/20 & 7/7	& 7/7 \\
\bf ApproximateEntropy & 20/20 & 20/20 & 20/20 & 20/20 & 7/7	& 7/7 \\
\bf RandomExcursions & 20/20 & 20/20 & 20/20 & 20/20 & 7/7	& 7/7 \\
\bf RandomExcursionsVariant & 20/20 & 20/20 & 20/20 & 20/20 & 7/7	& 7/7 \\
\bf Serial & 20/20 & 20/20 & 20/20 & 20/20 & 7/7	& 7/7 \\
\bf LinearComplexity & 20/20 & 20/20 & 20/20 & 20/20 & 7/7	& 7/7 \\
\hline \hline
\end{tabular}
%\arrayrulewidth=2pt
\label{tab:1}
\end{table*}

It is worth noting that samples to be tested were extracted continuously by days. Thus, interferences that may affect the experiments were more complex. In all scenarios, no special measures were adopted to reduce the interference of background light noise and temperature fluctuation. Despite a few failed tests, the results of the tests and analysis above indicate good quality of randomness of raw data from all of these three scenarios and postprocessing is not necessary. The results of Scenario (a) were consistent with the theory analysis and show the robust of our RNG from slowly varying interferences. Scenario (b) and Scenario (c), with double APD, suggest an APD array scheme, which is promising to break through the generation rate limitation. Scenario (c) also gave a relatively strict proof that dark count of APD is usable for RNG. The light source for this scheme, to be or not to be, is not a question any more.

It should be mentioned that although the RNG can generate high quality raw random bits, additional postprocessing methods \cite{Ma2013} still can be used to further improve the quality of the final output. The generic framework of randomness evaluating method and postprocessing algorithms proposed by reference \cite{Ma2013} provides an instructive guideline for design of the random signal extractors to achieve a tradeoff between the quality of randomness and the cost.

\section{Conclusion}
\label{Conclusion}

In conclusion, we have proposed and realized a robust and high-efficiency TRNG scheme. Dark counts of APD can be used as a resource in this scheme, so that a deep cooling process is not necessary and the system can work at RT. The fluctuation of pulse intensity arriving at the APD and slowly varying interferences affect only the efficiency of the RNG rather than the randomness of the RN series, so the scheme is compatible with light source and background photons. The experimental results also indicate the feasibility of integrating an APD array into a RNG chip, which can effectively increase the generation rate of RNs even uo to Gbps and make the scheme have more practical value.

\section*{Acknowledgements}\
\label{Acknowledgments}
This work has been supported by the National Basic Research Program of China (Grants No. 2011CBA00200 and No. 2011CB921200), the National Natural Science Foundation of China (Grant Nos. 61101137, 61201239, 61205118, 11304397). Fang-Xiang Wang and Chao Wang contributed equally to this work.

%\appendices
\section*{Appendix}
\label{Appendix}

\begin{lemma}

Define a function $f$ with expression below
\begin{displaymath}
f(N,k,p)=p^k(1-p)^{N-k}+p^{N-k}(1-p)^k
\end{displaymath}
Then, if $N\geq 2$ and $p\in (0,\ \frac{1}{2})\cup (\frac{1}{2},1)$, the equation about $k$
\begin{equation}
f(N,k,p)=f(N,k,\frac{1}{2})
\label{equ:1}
\end{equation}
has two and only two roots in $[0,N]$.
\label{lem:1}
\end{lemma}

\begin{proof}

First, we simplify equation (\ref{equ:1}) below
\begin{equation}
(1-p)^N\ (\frac{p}{1-p})^k+p^N\ (\frac{1-p}{p})^k=2\cdot(\frac{1}{2})^N.
\label{equ:2}
\end{equation}
The left of equation is symmetric about $p=\frac{1}{2}$, so we only need to consider the case $p\in (\frac{1}{2},1)$, and then $\frac{p}{1-p}>1$. Let $a=(p/(1-p))^k$. Then for any $k\in [0,N]$, we have $a\in [1,(\frac{p}{1-p})^N]$. As $k$ and $a$ are one to one, we simplify equation (\ref{equ:2}) and obtain
\begin{displaymath}
(1-p)^N\ a+p^N\ \frac{1}{a}=2\cdot \frac{1}{2}.
\end{displaymath}

Define the function
\begin{displaymath}
\varphi (a)=(1-p)^N\ a+p^N\ \frac{1}{a}-2\cdot \frac{1}{2}.
\end{displaymath}
as $p\in (\frac{1}{2},1)$, we have
\begin{displaymath}
\varphi (1)=\varphi [(\frac{p}{1-p})^N)]=(1-p)^N+p^N-2\cdot \frac{1}{2}>0.
\end{displaymath}
Next,the derivative of $\varphi$ is
\begin{displaymath}
\varphi' (a)=(1-p)^N)-p^N\ \frac{1}{a^2}.
\end{displaymath}
Let $\varphi' (a_0)=0$, $a_0\in (1,(\frac{p}{1-p})^N)$ we obtain $a_0=\sqrt{(\frac{p}{1-p})^N}$. Furthermore, if $a\in [1,\sqrt{(\frac{p}{1-p})^N})$, then $\varphi' (a)<0$;if  $a\in (\sqrt{(\frac{p}{1-p})^N},(\frac{p}{1-p})^N]$, then $\varphi' (a)>0$. In addition, $\varphi [(\frac{p}{1-p})^N)]=2\ \sqrt{p^N\ (1-p)^N)}-2\cdot \frac{1}{2}<0$, so the equation $\varphi (a)=0$ has two roots in $(1,{\frac{p}{1-p}}^N)$.

As $k$ and $a$ are one to one, it is easy to prove that equation (\ref{equ:1}) has two and only two roots in (0,N), which we denote as $ x_1 $ and $ x_2 $, $ x_1<x_2 $, then if $ k\in [0,x_1)\cup(x_2,N] $, $f(N,k,p)>f(N,k,\frac{1}{2})$; if $k\in (x_1,x_2)$, $f(N,k,p)<f(N,k,\frac{1}{2})$. In addition, as $f(N,k,p)=f(N,N-k,p)$, we have $x_1+x_2=N$.

This completes the proof of Lemma.

\end{proof}

\begin{theorem}

For any integer $N\geq 2$, the optimal $p$ for normalized extracted entropy $H(N,p)$ is $\frac{1}{2}$, and $H(N,p)\to S(p)$ as $N\to \infty$, where $S(p)$ is the Shannon entropy of a single Bernoulli trial, and $H(N,p)$ is defined as
\begin{displaymath}
H(N,p)=-\frac{1}{N}
\sum_{k=1}^{N-1}N_k\  p^k\ (1-p)^{N-k}\ (log_2\frac{1}{N_k}),
\end{displaymath}
and $N_k={N\choose k}$ is the binomial coefficient.
\label{theo:1}
\end{theorem}

\begin{proof}
Let
\begin{displaymath}
f(N,k,p)=p^k(1-p)^{N-k}+p^{N-k}(1-p)^k,
\end{displaymath}
then
\begin{equation}
\begin{aligned}
H(N,p)&=-\frac{1}{N}\sum_{k=1}^{N-1}N_k\  p^k\ (1-p)^{N-k}\ (log_2\frac{1}{N_k})\\
&=\frac{1}{N}
\sum_{k=0}^{N}N_k\  p^k\ (1-p)^{N-k}\ (log_2N_k)\\
&=\frac{1}{2N}
\sum_{k=0}^{N}[ p^k\ (1-p)^{N-k}\\
&\quad +p^{N-k}\ (1-p)^k]N_k (log_2N_k)\\
&=\frac{1}{2N}
\sum_{k=0}^{N}f(N,k,p)\ N_k (log_2N_k)
\end{aligned}
\label{equ:3}
\end{equation}
It is clear that when $p\in (0,\ \frac{1}{2})\cup (\frac{1}{2},\ 1)$,
\begin{displaymath}
\sum_{k=0}^{N}f(N,k,p)\ N_k=\sum_{k=0}^{N}f(N,k,\frac{1}{2})\ N_k=2,
\end{displaymath}
and $f(N,0,p)=(1-p)^N+p^N>2\cdot (\frac{1}{2})^N=f(N,\ 0,\ \frac{1}{2})$,
Thus, there exists an integer $k_0\in (0,N)$ such that $f(N,k_0,p)< f(N,k_0,\frac{1}{2})$. Additionally, from the proof of \textbf{Lemma \ref{lem:1}} we know there exist $x_1,x_2\in(0,N)$, $x_1<x_2$ and $x_1+x_2=N$ such that $f(N,x_1,p)- f(N,x_1,\frac{1}{2})=f(N,x_2,p)- f(N,x_2,\frac{1}{2})=0$, and if and only if $k\in(x_1,x_2)$, $f(N,k,p)<f(N,k,\frac{1}{2})$; thus, we have $k_0\in(x_1,x_2)$. In addition, if $x_1,x_2$ are integers, we have $k_0\in[x_1+1,x_2-1]$; otherwise $k_0\in[\lfloor x_1\rfloor+1,\lfloor x_2\rfloor]$, where $\lfloor x_1\rfloor$ represents the largest integer that is not larger than $x_1$.

With the conclusion above, we will show when $N\geq 2$ and $p\in(0,\ \frac{1}{2})\cup(\frac{1}{2},\ 1)$, we have $H(N,p)<H(N,1/2)$. There are two cases.\\
\textbf{Case 1:} $x_1,x_2$ are not integers, then
\begin{displaymath}
\begin{aligned}
&\quad 2\sum_{k=0}^{\lfloor x_1\rfloor}[f(N,k,p)-f(N,k,\frac{1}{2})]N_k\\
&=(\sum_{k=0}^{\lfloor x_1\rfloor}+\sum_{k=\lfloor x_2\rfloor +1}^{N})[f(N,k,p)-f(N,k,\frac{1}{2})]N_k\\
&=[2-\sum_{k=\lfloor x_1\rfloor +1}^{\lfloor x_2\rfloor}f(N,k,p)\ N_k]\\
&\quad -[2-\sum_{k=\lfloor x_1\rfloor +1}^{\lfloor x_2\rfloor}f(N,k,p)\ N_k]\\
&=\sum_{k=\lfloor x_1\rfloor +1}^{\lfloor x_2\rfloor}[f(N,k,\frac{1}{2})-f(N,k,p)]N_k.
\end{aligned}
\end{displaymath}
and
\begin{displaymath}
\begin{aligned}
&\quad\  H(N,p)-H(N,\frac{1}{2})\\
&=\sum_{k=0}^{N}[f(N,K,P)-f(N,k,\frac{1}{2})]\ \frac{N_k}{2N}\ log_2N_k\\
&=(\sum_{k=0}^{\lfloor x_1\rfloor}+\sum_{k=\lfloor x_2\rfloor +1}^{N})[f(N,k,p)-f(N,k,\frac{1}{2})]\ \frac{N_k}{2N}\ log_2N_k\\
&\quad +\sum_{k=\lfloor x_1\rfloor +1}^{\lfloor x_2\rfloor}[f(N,k,p)-f(N,k,\frac{1}{2})]\ \frac{N_k}{2N}\ log_2N_k\\
&=2\sum_{k=0}^{\lfloor x_1\rfloor}[f(N,k,p)-f(N,k,\frac{1}{2})]\ \frac{N_k}{2N}\ log_2N_k\\
&\quad +\sum_{k=\lfloor x_1\rfloor +1}^{\lfloor x_2\rfloor}[f(N,k,p)-f(N,k,\frac{1}{2})]\ \frac{N_k}{2N}\ log_2N_k\\
&\leq 2\sum_{k=0}^{\lfloor x_1\rfloor}[f(N,k,p)-f(N,k,\frac{1}{2})]\ \frac{N_k}{2N}\ log_2N_{\lfloor x_1\rfloor}\\
&\quad +\sum_{k=\lfloor x_1\rfloor +1}^{\lfloor x_2\rfloor}[f(N,k,p)-f(N,k,\frac{1}{2})]\ \frac{N_k}{2N}\ log_2N_{\lfloor x_1\rfloor+1}\\
&=\sum_{k=\lfloor x_1\rfloor +1}^{\lfloor x_2\rfloor}[f(N,k,p)-f(N,k,\frac{1}{2})]\ \cdot\frac{N_k}{2N}\ \\
&\qquad\cdot [log_2N_{\lfloor x_1\rfloor +1}-log_2N_{\lfloor x_1\rfloor}]\\
&<0.
\end{aligned}
\end{displaymath}
The last inequality holds, because there exists integer $k_0\in [\lfloor x_1\rfloor+1,\ \lfloor x_2\rfloor]$, and $\lfloor x_1\rfloor+1\leq \lfloor x_2\rfloor=N-1-\lfloor x_1\rfloor$, thus, $\lfloor x_1\rfloor+1\leq \frac{N}{2}$, so $N_{\lfloor x_1\rfloor+1}\geq N_{\lfloor x_1\rfloor}$.\\
\textbf{Case 2:} $x_1$,$x_2$ are integers, then analogously,
\begin{displaymath}
\begin{aligned}
&\quad\  H(N,p)-H(N,\frac{1}{2})\\
&=2\sum_{k=0}^{x_1}[f(N,k,p)-f(N,k,\frac{1}{2})]\ \frac{N_k}{2N}\ log_2N_k\\
&\quad +\sum_{k=x_1+1}^{x_2-1}[f(N,k,p)-f(N,k,\frac{1}{2})]\ \frac{N_k}{2N}\ log_2N_k\\
&\leq 2\sum_{k=0}^{x_1}[f(N,k,p)-f(N,k,\frac{1}{2})]\ \frac{N_k}{2N}\ log_2N_{x_1}\\
&\quad +\sum_{k=x_1+1}^{x_2}[f(N,k,p)-f(N,k,\frac{1}{2})]\ \frac{N_k}{2N}\ log_2N_{x_1+1}\\
&=\sum_{k=x_1+1}^{x_2}[f(N,k,p)-f(N,k,\frac{1}{2})]\cdot\frac{N_k}{2N}\\
&\quad\cdot [log_2N_{x_1+1}-log_2N_{x_1}]\\
&<0.
\end{aligned}
\end{displaymath}

The last inequality holds, because there exists integer $k_0\in [x_1+1,x_2-1]$, and $x_1+1\leq x_2-1=N-x_1-1$, thus $x_1+1\leq \frac{N}{2}$, so $N_{x_1+1}>N_{x_1} $.
We have therefore proved that the optimal $p$ for normalized extracted entropy $H(N,p)$ is $\frac{1}{2}$, and we will next show the remaining part.

First, as
\begin{displaymath}
2^N=\sum_{k=0}^{N}{N\choose k}=\sum_{k=0}^{N}N_k.
\end{displaymath}
we have $\frac{log_2 N_k}{N}<1$ for $0\leq k\leq N$. Assuming that $0<p<\frac{1}{2}$, the cases that $\frac{1}{2}<p<1$ and $p=\frac{1}{2}$ are similar, so there exists $\delta >0$ sufficiently small such that $p+\delta <\frac{1}{2}$, by the weak law for a binomial distribution,
\begin{displaymath}
\lim_{N\to \infty}\sum_{p-\delta <\frac{k}{N}<p+\delta}N_k\ p^k\ (1-p)^{N-k}=1.
\end{displaymath}
Thus, given any $\epsilon>0$, there is an $\mathcal{N_0}$ such that for $N>\mathcal{N_0}$,
\begin{equation}
\sum_{\mid\frac{k}{N}-p\mid\geq \delta}N_k\ p^k\ (1-p)^{N-k}<\epsilon\ .
\label{equ:4}
\end{equation}
Thus, together with $\frac{log_2 {N_k}}{N}<1$, we have
\begin{displaymath}
\begin{aligned}
&\quad \sum_{\mid\frac{k}{N}-p\mid<\delta}N_k\ p^k\ (1-p)^{N-k}\frac{log_2 {N_k}}{N}\\
&<H(N,p)\\
&<\sum_{\mid\frac{k}{N}-p\mid<\delta}N_k\ p^k\ (1-p)^{N-k}\frac{log_2 {N_k}}{N}+\epsilon\ .
\end{aligned}
\end{displaymath}
Because $p+\delta<\frac{1}{2}$, when $\mid\frac{k}{N}-p\mid<\delta$, we have $log_2 N_{\lfloor N(p-\delta)\rfloor-1}\leq log_2N_k\leq log_2N_{\lfloor N(p+\delta)\rfloor+1}$, so
\begin{displaymath}
\centering
\begin{aligned}
&\quad \sum_{\mid\frac{k}{N}-p\mid<\delta}N_k\ p^k\ (1-p)^{N-k}\frac{log_2 N_{\lfloor N(p-\delta)\rfloor-1}}{N}\\
&<H(N,p)\\
&<\sum_{\mid\frac{k}{N}-p\mid<\delta}N_k\ p^k\ (1-p)^{N-k}\frac{log_2 N_{\lfloor N(p+\delta)\rfloor+1}}{N}+\epsilon\ .
\end{aligned}
\end{displaymath}
Together with equation (\ref{equ:4}), we get
\begin{displaymath}
\begin{aligned}
&\quad (1-\epsilon)\frac{log_2 N_{\lfloor N(p-\delta)\rfloor-1}}{N}\\
&<H(N,p)\\
&<\frac{log_2 N_{\lfloor N(p+\delta)\rfloor+1}}{N}+\epsilon\ .
\end{aligned}
\end{displaymath}
Using Stirling's formula on both side
\begin{displaymath}
(1-\epsilon)S(p-\delta)\leq \lim_{N\to \infty}H(N,p)\leq S(p+\delta)+\epsilon\ .
\end{displaymath}
As $\epsilon$ and $\delta$ are arbitrary, with the continuity of $S(p)$, we obtain $\lim_{N\to \infty}H(N,p)=S(p)$.

This completes the proof of \textbf{Theorem \ref{theo:1}}.
\end{proof}

\begin{theorem}
$H(N,\frac{1}{2})$ increases to $1$ as $N$ approaches infinity.
\label{theo:2}
\end{theorem}

\begin{proof}
With the conclusion of \textbf{Theorem \ref{theo:1}}, we have $H(N,\frac{1}{2})$ converging to 1 as N approaches infinity. It therefore remains to be proven that $H(N,\frac{1}{2})$ is an increasing function.

First, we have
\begin{displaymath}
\begin{aligned}
H(N,\frac{1}{2})&=\frac{1}{(N+1)\cdot 2^{N+1}}\sum\limits_{k=1}^{N}(N+1)_k\ log_2(N+1)_k\\
&=\frac{1}{(N+1)\cdot 2^{N+1}}[\sum\limits_{k=1}^{N}N_k\ log_2(N_k\frac{N+1}{N+1-k})\\
&\quad +\sum\limits_{k=1}^{N}N_{k-1}\ log_2(N_{k-1}\frac{N+1}{k})]\\
&=\frac{1}{(N+1)\cdot 2^{N+1}}[\sum\limits_{k=0}^{N}N_k\ log_2(N_k\frac{N+1}{N+1-k})\\
&\quad +\sum\limits_{k=0}^{N}N_k\ log_2(N_k\frac{N+1}{k+1})]\\
&=\frac{N}{N+1}H(N,\frac{1}{2})+\frac{1}{(N+1)2^{N+1}}\\
&\quad \cdot [\sum\limits_{k=0}^{N}N_k\ log_2\frac{(N+1)^2}{(N+1-k)(k+1)}]\ .
\end{aligned}
\end{displaymath}
Thus, $H(N+1,\frac{1}{2})\geq H(N,\frac{1}{2})$ is equivalent to
\begin{displaymath}
\begin{aligned}
&\qquad\quad\frac{1}{2^{N+1}}[\sum\limits_{k=0}^{N}N_k\ log_2\frac{(N+1)^2}{(N+1-k)(k+1)}]\\
&\qquad\ \geq H(N,\frac{1}{2})=\frac{1}{N\cdot 2^N}\sum\limits_{k=0}^{N}N_k\ log_2N_k\\
&\Longleftrightarrow\quad N[\sum\limits_{k=0}^{N}N_k\ log_2\frac{(N+1)^2}{(N+1-k)(k+1)}]\\
&\qquad\ \geq 2\sum\limits_{k=0}^{N}N_k\ log_2N_k\\
&\Longleftrightarrow\quad \sum\limits_{k=0}^{N}N_k\ log_2\frac{(N+1)^{2N}}{[(N+1-k)(k+1)]^N(N_k)^2}\geq 0\ .
\end{aligned}
\end{displaymath}
As $(N+1-k)(k+1)\leq(\frac{N}{2}+1)^2$ and $N_k\leq N_{\lfloor \frac{N}{2}\rfloor}$, it is sufficient to prove
\begin{equation}
(\frac{2N+2}{N+2})^N\geq N_{\lfloor \frac{N}{2}\rfloor}\ .
\label{inequ:5}
\end{equation}

Next, we prove inequality (\ref{inequ:5}) by induction. When $N=1$, $(\frac{2N+2}{N+2})^N=\frac{4}{3}\geq 1=N_{\lfloor \frac{N}{2}\rfloor}$;
When $N=2$, $(\frac{2N+2}{N+2})^N=\frac{9}{4}\geq 2=N_{\lfloor \frac{N}{2}\rfloor}$. Assuming that when $N=k-1$, $k\geq 2$, we have $(\frac{2k}{k+1})^(k-1)\geq (k-1)_{\lfloor \frac{k-1}{2}\rfloor}$. Then, when $N=k+1$, we have $(k+1)_{\lfloor\frac{k+1}{2}\rfloor}=\frac{k(k+1)}{\lfloor\frac{k+1}{2}\rfloor (k+1-\lfloor\frac{k+1}{2}\rfloor)}\cdot (k-1)_{\lfloor \frac{k-1}{2}\rfloor}$. By the induction hypothesis,$(k+1)_{\lfloor\frac{k+1}{2}\rfloor}\leq \frac{k(k+1)}{\lfloor\frac{k+1}{2}\rfloor (k+1-\lfloor\frac{k+1}{2}\rfloor)}(\frac{2k}{k+1})^(k-1)$, it remains to show
\begin{equation}
\frac{k(k+1)}{\lfloor\frac{k+1}{2}\rfloor (k+1-\lfloor\frac{k+1}{2}\rfloor)}(\frac{2k}{k+1})^(k-1)\leq (\frac{2k+4}{k+3})^{k+1}.
\label{inequ:6}
\end{equation}
We prove the inequality (\ref{inequ:6}) by two cases.\\
\textbf{Case 1:} $k$ is even, then $\lfloor\frac{k+1}{2}\rfloor=\frac{k}{2}$, and inequality (\ref{inequ:6}) is equivalent to
\begin{equation}
\begin{aligned}
&\qquad\ \frac{k(k+1)}{\frac{k}{2}(k+1-\frac{k}{2})}(\frac{2k}{k+1})^{k-1}\leq (\frac{2k+4}{k+3})^{k+1}\\
&\Longleftrightarrow \frac{k+1}{k+2}(\frac{k}{k+1})^{k-1}\leq (\frac{k+2}{k+3})^{k+1}\\
&\Longleftrightarrow (\frac{k(k+3)}{(k+1)(k+2)})^{k-1}\leq \frac{(k+2)^3}{(k+1)(k+3)^2}\\
&\Longleftrightarrow (\frac{k^2+k3}{k^2+3k+2})^{k-1}\leq \frac{(k+2)^3}{(k+1)(k+3)^2}.
\end{aligned}
\label{inequ:7}
\end{equation}
Since
\begin{displaymath}
\begin{aligned}
&\ (\frac{k^2+k3}{k^2+3k+2})^{k-1}\\
&\leq \frac{k^2+3k}{k^2+3k+2}\cdot \frac{k^2+3k+1}{k^2+3k+3}\cdots \frac{k^2+4k-2}{k^2+4k}\\
&=\frac{(k^2+3k)(k^+3k+1)}{(k^2+4k-1)(k^2+4k)}\\
&=\frac{(k+3)(k^+3k+1)}{(k+4)(k^2+4k-1)},
\end{aligned}
\end{displaymath}
it remains to show
\begin{displaymath}
\begin{aligned}
&\qquad\quad\frac{(k+3)(k^+3k+1)}{(k+4)(k^2+4k-1)}\leq \frac{(k+2)^3}{(k+1)(k+3)^2}\\
&\Longleftrightarrow\ (k+3)^3(k+1)(k^2+3k+1)\\
&\qquad\leq (k+2)^3(k+4)(k^2+4k-1)\\
&\Longleftrightarrow\ (k^4+10k^3+36k^2+54k+27)(k^2+3k+1)\\
&\qquad\leq (k^4+10k^3+36k^2+56k+32)(k^2+4k-1).
\end{aligned}
\end{displaymath}
Since $k\geq 2$, we have $k^2+3k+1\leq k^2+4k-1$, so the inequality above is hold.\\
\textbf{Case 2:} $k$ is odd, then $\lfloor\frac{k+1}{2}\rfloor=\frac{k+1}{2}$, and inequality (\ref{inequ:6}) is equivalent to
\begin{displaymath}
\begin{aligned}
&\qquad\quad\frac{k(k+1)}{(\frac{k+1}{2})^2}(\frac{2k}{k+1})^{k-1}\leq (\frac{2k+4}{k+3})^{k+1}\\
&\Longleftrightarrow\ (\frac{k}{k+1})^k\leq (\frac{k+2}{k+3})^{k+1}.
\end{aligned}
\end{displaymath}
Since $\frac{k}{k+1}\leq\frac{k+1}{k+2}$, it is sufficient to prove $\frac{k+1}{k+2}(\frac{k}{k+1})^{k-1}\leq (\frac{k+2}{k+3})^{k+1}$. Notice that the inequality above is the same as inequality (\ref{inequ:7}), so the next proof is the same in \textbf{Case 1}.

This completes the proof of \textbf{Theorem \ref{theo:2}}.

\end{proof}

%\end{comment}

%\section*{Acknowledgment}
%Fang-Xiang Wang and Chao Wang contributed equally to this work.

% Can use something like this to put references on a page
% by themselves when using endfloat and the captionsoff option.

% trigger a \newpage just before the given reference
% number - used to balance the columns on the last page
% adjust value as needed - may need to be readjusted if
% the document is modified later
%\IEEEtriggeratref{8}
% The "triggered" command can be changed if desired:
%\IEEEtriggercmd{\enlargethispage{-5in}}

% references section

% can use a bibliography generated by BibTeX as a .bbl file
% BibTeX documentation can be easily obtained at:
% http://www.ctan.org/tex-archive/biblio/bibtex/contrib/doc/
% The IEEEtran BibTeX style support page is at:
% http://www.michaelshell.org/tex/ieeetran/bibtex/
%\bibliographystyle{IEEEtran}
% argument is your BibTeX string definitions and bibliography database(s)
%\bibliography{IEEEabrv,../bib/paper}
%
% <OR> manually copy in the resultant .bbl file
% set second argument of \begin to the number of references
% (used to reserve space for the reference number labels box)

\end{document}